\documentclass[reprint, amsmath,amssymb, prl, superscriptaddress]{revtex4-1}
\usepackage{mathtools}
\DeclarePairedDelimiterX{\infdivx}[2]{(}{)}{%
  #1\;\delimsize\|\;#2%
}

\usepackage{color}
\usepackage{graphicx}% Include figure files
\usepackage{dcolumn}% Align table columns on decimal point
\usepackage{bm}% bold math
\usepackage{physics}
\usepackage{hyperref}
\usepackage{subcaption}
\usepackage{amsmath}
\usepackage{ dsfont }
\usepackage{verbatim}
\usepackage{amssymb}
\usepackage{amsfonts}              % for blackboard bold, etc
\usepackage{amsthm}                % better theorem environments

\usepackage{kotex}
\usepackage{qcircuit}
\usepackage[normalem]{ulem}

\newtheorem{thm}{Theorem}

\newtheorem{lem}{Lemma}
\newtheorem{fac}{Fact}
\newtheorem{con}{Conjecture}

\newtheorem{theorem}{Theorem}

\begin{document}

\title{Randomness cost of masking quantum information \\and the information conservation law}

\author{Seok Hyung Lie}
\affiliation{
 Department of Physics and Astronomy, Seoul National University, Seoul, 151-742, Korea
}%
\author{Hyunseok Jeong}
\affiliation{
Department of Physics and Astronomy, Seoul National University, Seoul, 151-742, Korea
}%

\date{\today}

\begin{abstract}
  Masking quantum information, which is impossible without randomness as a resource, is a task that encodes quantum information into bipartite quantum state while forbidding local parties from accessing  to that information. In this work, we disprove the geometric conjecture about unitarily maskable states [K. Modi et al., Phys. Rev. Lett. {\bf 120}, 230501 (2018)], and make an algebraic analysis of quantum masking. First, we show a general result on quantum channel mixing that a subchannel's mixing probability should be suppressed if its classical capacity is larger than the mixed channel's capacity. This constraint combined with the well-known  \textit{information conservation law}, a law that does not exist in classical information theory, gives a lower bound of randomness cost of masking quantum information as a monotone decreasing function of evenness of information distribution. This result provides a consistency test for various scenarios of fast scrambling conjecture on the black hole evaporation process. The results given here are robust to incompleteness of quantum masking.
  
\end{abstract}

\pacs{Valid PACS appear here}
\maketitle
    How can one hide information  from two parties holding its containers? Na\"{i}vely, one can tear the piece of paper containing the information into two pieces and distribute them to two parties so that it remains recoverable when the pieces are gathered together at later point of time. However, this method leaks some amount of information to each party. To hide $n$ bits of classical information completely, one needs to get $n$ bits of classical information (`randomness') maximally correlated with it, for example, by using one-time pad cipher \cite{shannon1949communication}. Is it possible to hide $n$ qubits of quantum information by making $n$ qubits quantumly correlated with it? The no-hiding theorem and the no-masking theorem answer this question negatively for pure state case \cite{braunstein2007quantum,modi2018masking}. In our previous work, we showed that one still needs additional $n$ bits of additional randomness to hide $n$ qubits of quantum information \cite{lie2019unconditionally}. Results of \cite{lie2019unconditionally} imply that different types of quantum correlation allowed in quantum masking processes require different amount of randomness consumption, but exact relation was still an open problem.

In this work, we illuminate an information theoretical reason behind these phenomena and specify the backbone of the mechanism; the \textit{information conservation law} of quantum mechanics. The information theoretical investigation yields a lower bound of the amount of randomness required for information masking task in terms of a measure of how unevenly information is distributed between the system and the environment. This suggests that important quantity for determining the minimal randomness consumption is not the exact amount of quantum correlation allowed between two parties but the unevenness of information distribution between two parties. Our result has implications for quantum secret sharing protocols and a class of proposals for resolution of the black hole information paradox called \textit{scrambling} \cite{hayden2007black}. 

\textbf{Introduction} Let $\mathcal{B}(\mathcal{H})$ be the Banach space of bounded operators on the Hilbert space $\mathcal{H}$ and $\mathcal{S}(\mathcal{H})$ be the set of quantum states on the Hilbert space $\mathcal{H}$.
A \textit{universal masking process} or \textit{universal quantum masker} \cite{modi2018masking,lie2019unconditionally} $\Phi_M : \mathcal{B}(\mathcal{H}_I) \to \mathcal{B}(\mathcal{H}_A \otimes \mathcal{H}_B)$ is an invertible quantum map that has constant marginal states,
\begin{equation}
    \forall \rho \in \mathcal{S}(\mathcal{H}), \;\; \Tr_B\Phi_M(\rho)=\sigma_A \;\;\text{and}\;\; \Tr_A\Phi_M(\rho)=\sigma_B,
\end{equation}
namely, a quantum masker distributes quantum state to two parties so that each party has no access to any information about the original quantum state. If it cannot mask all the states in $\mathcal{S(H)}$, then it is not universal.  Its invertibility allows the following expression \cite{nayak2007invertible},
\begin{equation} \label{rev}
    \Phi_M(\rho)=M(\rho_I \otimes \sigma_S)M^\dag.
\end{equation}
with some isometry $M:\mathcal{H}_I\otimes \mathcal{H}_S \to \mathcal{H}_A \otimes \mathcal{H}_B$ and a quantum state $\sigma_S$. We call the state $\sigma_S$ the \textit{safe state} of the masking process and we interpret it as a quasiclassical randomness source that is needed to mask the quantum information. Our interest is to investigate the relation between the\textit{ randomness cost }of masking process, that is, for any pure state $\ket{\phi} \in \mathcal{H_I}$,
\begin{equation}
    \mathcal{R}(\Phi_M):= S(\Phi_M(\dyad{\phi}_I)) = S(\sigma_S),
\end{equation}
where $S$ is the von Neumann entropy, and the characteristics of quantum interaction $M$ of the process. The following result on the size of each party is known \cite{gottesman2000theory, imai2005information, lie2019unconditionally}.

\begin{fac}[The generalized quantum masking theorem]
For a universal quantum masker $\Phi_M : \mathcal{B}(\mathcal{H}_I) \to \mathcal{B}(\mathcal{H}_A \otimes \mathcal{H}_B)$ with constant marginal states $\sigma_A$ and $\sigma_B$ for systems $A$ and $B$ and the safe state $\sigma_S$, we have
$$\min\{S(\sigma_A),S(\sigma_B),S(\sigma_S)\}\geq \log d.$$
\end{fac}

If $\sigma_S=\sum_i p_i \dyad{i}_S$ is the spectral decomposition of $\sigma_S$, then $\Phi_M$ can be expressed as a random isometry operation,

\begin{equation}
    \Phi_M(\rho)=\sum_i p_i M_i \rho M_i^\dag,
\end{equation}
where $M_i:=M (\mathds{1}_I \otimes \ket{i}_S)$ are isometries from $\mathcal{H}_I$ to $\mathcal{H}_A \otimes \mathcal{H}_B$. We will call such isometry a \textit{bipartite embedding}. These isometries also have orthogonal images ($M_i^\dag M_j =0$ if $i \neq j)$. Such decomposition of $\Phi_M$ is unique up to degeneracy of spectrum of the safe state $\sigma_S$.

In quantum information theory, there is a unique conservation law of information in contrast with classical information theory. When an entangled state $\ket{\Psi}_{RI}$ undergoes an arbitrary isometry $U:\mathcal{H}_I \to \mathcal{H}_A \otimes \mathcal{H}_B$, then the system $R$'s mutual information with two output systems $A$ and $B$ are conserved in the following sense.

\begin{lem}[Information Conservation Law]
\begin{equation} \label{eqn:infocon}
    2S(R) = I(R:A) + I(R:B).
\end{equation}
\end{lem}
 Note that this simple law alone can derive both the \textit{no-hiding theorem} \cite{braunstein2007quantum} and the \textit{no-masking theorem} \cite{modi2018masking}. If the bipartite embedding $U:\mathcal{H}_I \to \mathcal{H}_A \otimes \mathcal{H}_B$ is a hiding or masking process so that output system $A$ is in a constant quantum state regardless of input state of the system $I$, then for a maximally entangled state $\ket{\Gamma}_{RI}$ with $\dim (\mathcal{H}_I)=\dim(\mathcal{H}_R)=d$, we have $I(R:A)=0$ as the systems $R$ and $A$ are in a product state. But this immediately implies $I(R:B)=2\log d >0$, i.e., the system $B$ is maximally entangled with $R$. This implies that there is no bipartite embedding that can hide quantum information from both parties. From this observation we can expect the information conservation law could give an insight on the generalized quantum masking theorem.

We can interpret $I(R:A)$ as the information flow from $I$ to $A$ when $\Tr_B (U \cdot U^\dag)$ is considered a channel from $I$ to $A$. We can intuitively anticipate that if too much information has flowed to a single system, then a large amount of randomness is needed to `scramble' to mask that information, just as it is for the classical one-time pad cipher. In the following section, we first prove an easily applicable result that verifies this intuition from simple entropic properties, and again prove a stronger theorem from information-theoretic argument.

\textbf{Main results} From now on, we will fix a $d-$dimensional universal quantum masker $\Phi_M$ and its bipartite embedding decomposition $\Phi_M(\rho) = \sum_i p_i M_i \rho M_i^\dag$. Also every entropic quantities with subscript $i$ refers to the corresponding quantity for the pure state $(\mathds{1}_R \otimes M_i)\ket{\Gamma}_{RI}$ of the tripartite system $RAB$ and unindexed entropic quantities such as $S(X)$  are the corresponding values for the system $X$ in the state $(\mathds{1}_R \otimes \Phi_M)(\dyad{\Gamma}_{RI})$ of the tripartite system $RAB$. For every quantum channel $\mathcal{N}$ into $\mathcal{H}_A\otimes\mathcal{H}_B$, we will denote its partial traces as $\Tr_B \circ \mathcal{N} = \mathcal{N}^A$ and $\Tr_A \circ \mathcal{N} = \mathcal{N}^B$.

We first investigate the power of isometry quantum masker by disproving the conjecture given in \cite{modi2018masking}.

\begin{con}[Modi et al. \cite{modi2018masking}] \label{thm:conj}
 For every isometry quantum masker $\mathcal{M}(\cdot) = M \cdot M^\dag :\mathcal{B(H_I)}\to\mathcal{B}(\mathcal{H}_A)\otimes\mathcal{B}(\mathcal{H}_B)$, its set of maskable states is in a `disk', the convex hull of a set of states $\{\dyad{\psi}:\ket{\psi}=\sum_{k=1}^d r_k e^{i\theta_k}\ket{k}, \theta_k\in [-\pi,\pi] \}$ with fixed non-negative real numbers $r_k$ such that  $\sum_k r_k^2 = 1$ and a fixed orthonormal basis $\{\ket{k}\}_{k=1}^d$.
\end{con}

We give a counterexample that disproves this conjecture. Consider a qudit-qudit system $\mathcal{H_I}$, that is, $\dim ( \mathcal{H_I}) = d^2$, and a masking process which is simply distributing a qudit to each party. Note that this masking process is essentially equivalent to any masking process from $d^2$-dimensional system to $d$-dimensional two parties since any $d^2-$dimensional unitary $M$ applied before the distribution only amounts to a change of basis.

Now, consider the set $S$ of bipartite states that has maximally mixed states as its marginal states. This is a set of maskable states of the masking process given above. As in Conjecture 1, suppose that there exists a basis $\{\ket{M_i}\}_{i=1}^{d^2}$ of $\mathcal{H_I}$ such that every state in $S$ has the same diagonal element with respect to this basis. It implies that for any local basis $\{\ket{a_i}\}_{i=1}^d$ and $\{\ket{b_i}\}_{i=1}^d$ of $\mathcal{H_A}$ and $\mathcal{H_B}$ respectively, following constraint is required with varying phases $\theta_j \in [-\pi,\pi]$ for any $1 \leq j \leq d$,

\begin{equation} \label{eqn:const}
   \frac{1}{d}\abs{\sum_j e^{i\theta_j} \bra{M_k}(\ket{a_j}\otimes\ket{b_j})}^2 = const. ,
\end{equation} 
for every $1 \leq k \leq d^2$. It follows from the fact that any maximally entangled state has maximally mixed states as marginals states. If $\ket{M_k}$ is not a product state, then by picking the Schmidt bases of $\ket{M_k}$ as local bases one can violate the constraint above. Therefore $\{\ket{M_k}\}_{k=1}^{d^2}$ is a product basis of the form $\ket{M_{d(i-1)+j}}=\ket{\alpha_i}_A\ket{\beta_j}_B$ for some local bases $\{\ket{\alpha_i}\}_{i=1}^d$ and $\{\ket{\beta_i}\}_{i=1}^d$. However, requiring the constraint above again for the discrete Fourier transformed local bases $\{\ket{\tilde{\alpha}_j}:=\sum_k \frac{1}{\sqrt{d}}e^{i 2\pi jk/d}\ket{\alpha_k}\}_{j=1}^d$ (and similarly defined $\{\ket{\tilde{b}_j}\}_{j=1}^d$) leads to the violation of the constraint. Therefore there exists no such basis $\{\ket{M_k}\}_{k=1}^{d^2}$.

%The constraint above is impossible unless the value above vanishes regardless of $\theta_j$ or each $\ket{M_k}$ is a product state of the form $\ket{M_k}=\ket{a_{i(k)}}\ket{b_{j(k)}}$. As one can shift the index of basis elements (e.g. $\ket{a_j}\to\ket{a_{j\oplus n}}$) and $\{\ket{a_i} \otimes \ket{b_j}\}_{i,j=1}^d$ forms a complete basis of $\mathcal{H}_A \otimes \mathcal{H}_B$,  by requiring the condition after each shifting, $\ket{M_k}$ should be written in a product state form with respect to some basis $\{\ket{a_i}\}_{i=1}^d$ and $\{\ket{b_j}\}_{j=1}^d$.  However, by picking, say, the discrete Fourier transformed basis $\{\ket{\tilde{a}_j}:=\sum_k \frac{1}{\sqrt{d}}e^{i 2\pi jk/d}\ket{a_k}\}_{j=1}^d$ of the given $\{\ket{a_i}\}_{i=1}^d$ (and similarly for $\{\ket{\tilde{b}_j}\}_{j=1}^d$) and demanding the same condition of $(\ref{eqn:const})$, we end up in a contradiction that $\ket{M_i}$ has at least two nonequivalent expressions as product state. Therefore there is no such basis $\{\ket{M_i}\}_{i=1}^{d^2}$.

This suggests that information hidden by a quantum masking process exploiting inherent multipartite structure of quantum state (e.g. interpreting 4-level system as a two-qubit system) need not be limited to the phase information with respect to a certain `classical' basis. In fact, any kind of correlation between two subsystems (be it quantum or classical) could be hidden when these two subsystems are separated.

The information conservation law suggests that at least $\log d$ bits of information should leak to at least one party, but our counterexample suggests that $\log d$ bits of information need not correspond to the classical information with respect to some preferred basis. Indeed, the information could correspond to quantum information of $\sqrt{d}$ dimensional quantum system. We are left in the situation in which there is no convenient geometrical tool to analyze the masking power of each isometry quantum masker, i.e., bipartite embedding. Thus the information conservation law $(\ref{eqn:infocon})$ is the best principle one could rely on.  By using some simple entropic properties, we can first derive the following preliminary result about quantum masking processes consuming randomness. Proofs of the theorems in the following sections can be found in the appendix at the end of the article.

\begin{theorem} \label{thm:naive}
For any universal quantum masker $\Phi_M$,
\begin{equation*} \label{naive}
    \max_{X\in\{A,B\}}\sum_i p_i I(R:X)_i \leq \mathcal{R}(\Phi_M).
\end{equation*}

\end{theorem}

\iffalse

\begin{proof}
Without loss of generality, we can assume $X=A$. Then as $I(R:A)_i = S(R)_i + S(A)_i - S(B)_i = S(R) + S(A)_i - S(B)_i$, it suffices to prove the following inequality,

\begin{equation}
    S(R) + \sum_i p_i S(A)_i \leq H(\{p_i\}) + \sum_i p_i S(B)_i,
\end{equation}
because $H(\{p_i\})=S(\sigma_S)$, where $H(\{q_i\}):= -\sum_i q_i \log q_i$ is the Shannon entropy of the probability distribution $\{q_i\}$. As $\sum_i p_i S(\rho_i) \leq S(\sum_i p_i \rho_i)$ from the concavity of von Neumann entropy \cite{watrous2018theory}, we have 
\begin{equation}
    S(R) + \sum_i p_i S(A)_i \leq S(R) + S(A) = S(RA),
\end{equation}
because systems $R$ and $A$ are in a product state because of the masking property of $\Phi_M$. If we define $\Phi_{i}^A(\rho) := \Tr_B \big(M(\rho_I\otimes \dyad{i}_S)M^\dag \big)$, $S(RA)$ equals to
\begin{equation}
    \begin{gathered}
    S(\sum_i p_i (\mathds{1}_R\otimes \Phi_{i}^A)(\dyad{\Gamma}_{RI})).
    \end{gathered}
\end{equation}
From the following property of von Neumann entropy \cite{watrous2018theory},
\begin{equation}
    S(\sum_i q_i \rho_i) \leq H(\{q_i\}) + \sum_i q_i S(\rho_i),
\end{equation}
we have
\begin{equation}
    S(RA) \leq H(\{p_i\}) + \sum_i p_i S(RA)_i.
\end{equation}
 This proves the wanted inequality since $S(RA)_i = S(B)_i$ because the systems $RAB$ are in a pure state. Since this inequality holds for arbitrary maximally entangled state $\ket{\Gamma}_{RI}$, we have the wanted result.
\end{proof}

\fi

This result already implies the generalized quantum masking theorem since $\max_{X\in\{A,B\}}\sum_i p_i I(R:X)_i \geq \log d$ as $\sum_i p_i I(R:A)_i + \sum_i p_i I(R:B)_i = 2 \log d$ from the information conservation law.

However, often the lower bound above is not enough for many quantum maskers.  A more detailed result can be obtained by observing that quantum masking process consuming randomness can be considered a channel mixing process. We can derive the following trade-off relation between the channel capacity of a subchannel and its mixing probability.

\begin{theorem} \label{thm:cap}
    Let a convex sum of subchannels $\{\mathcal{N}_i\}$ $\sum_i p_i \mathcal{N}_i = \mathcal{N}$ be an almost-erasure quantum channel, i.e., a quantum map that has entanglement-assisted classical capacity of $e$. Then for all $i$, the entanglement-assisted classical capacity $C_{EA}$ of each masking component is upper bounded by the information content of the randomness source for its corresponding $i$, i.e.,
    \begin{equation} \label{ineq1}
        C_{EA}(\mathcal{N}_i) - e \leq -\log p_i,
    \end{equation}
\end{theorem}

Note that actually the proof of the theorem above can be applied to the classical capacity of quantum channel with any kind of proper resource assumption, not necessarily the unboundedness of pre-distributed entanglement.

From the fact \cite{bennett1999entanglement,bennett2002entanglement} that for any quantum channel $\mathcal{N}:A' \to B$,
\begin{equation}
    \max_{\phi_{AA'}} I(A:B)_{\tau_{AB}} = C_{EA}(\mathcal{N}),
\end{equation}
where $\phi_{AA'}$ is a pure state on $AA'$ and $\tau_{AB} = (\mathds{1}_A \otimes \mathcal{N}_{A' \to B})(\phi_{AA'})$, we have the following corollary
\begin{equation} \label{ineq2}
    \max\{I(R:A)_{\tau_{RA}},I(R:B)_{\tau_{RB}}\} \leq -\log p_i,
\end{equation}
for an arbitrarily given bipartite pure state $\phi_{RI}$ with $\tau_{RAB}=(\mathds{1}_R \otimes \Phi_i)(\phi_{RI})$. Choosing arbitrary maximally entangled state $\phi_{RI}$ and averaging both sides of (\ref{ineq2}) leads to the following result.

\begin{theorem} \label{thm:inq} For any $d-$dimensional universal quantum masker $\Phi_M$ with with the safe state with spectrum $\{p_i\}$, the following inequality holds.
\begin{equation} \label{eqn:mid}
    \log d + \sum_i p_i |S(A)_i - S(B)_i| \leq \mathcal{R}(\Phi_M).
\end{equation}
\end{theorem}

 Theorem \ref{thm:inq} can be considered a combination of two results. First, higher information influxes should be scrambled with larger amount of randomness if their net influx should be suppressed under a certain value. Second, information cannot be destroyed or hidden under unitary interaction; it either flows to the system or to the system. Therefore, a set of unitary interaction that allows information flow to each party to be more even requires less amount of randomness to form a quantum masking process.

It is worth defining a measure of evenness of information distribution between two parties that only depends on the set of bipartite embeddings. Consider a measure defined in the following way with the notation $I_i := \max \{I(R:A)_i,I(R:B)_i\}$, 

\begin{equation} \label{eqn:div}
    \mathcal{I}_{1} (\{M_i\}_{i \in I}) := \max_{\ket{\Gamma}} \min_{S \subseteq I} H\left(\left\{\frac{1}{2^{I_i}}\right\}_{i\in S}\right),  
\end{equation}
where $H(\{t_i\}_{i \in T}):= -\sum_{i\in T} t_i \log t_i$ is formally defined as the Shannon entropy even for the set of nonnegative numbers $\{a_i\}_{i \in T}$ that is not a probability distribution. The maximization is over the choice of initial bipartite pure state $\ket{\Gamma}_{RI}$ and the minimization is over the subset of indices $S \subseteq I$ such that $\sum_{i\in S} 2^{-I_i} \geq 1$ with existence of $i_0 \in S$ such that $\sum_{i \in S} 2^{-I_i} - 2^{-I_{i_0}} \leq 1$. Since  $\mathcal{I}_1$ is monotone increasing function of $\max\{I(R:A)_i,I(R:B)_i\}$ for each $i$, $\mathcal{I}_1$ is a legitimate measure of information unevenness. $\mathcal{I}_1$ is bounded as $\log d \leq \mathcal{I}_1 \leq 2(1+d^{-1}) \log d$.

 Once the one-shot measure $\mathcal{I}_1$ is defined, one can define its regularized version $\mathcal{I}_\infty(\{M_i\}):= \lim_{n\to \infty} \frac{1}{n} \mathcal{I}_1(\{M_i\}^{\otimes n})$. This regularized measure has the bound of $\log d \leq \mathcal{I}_\infty \leq 2\log d$. We have the following lower bound of the randomness cost of quantum masking process that is monotone decreasing function of evenness of information distribution thereof. 

\iffalse

Both $\mathcal{I}_1$ and $\mathcal{I}_\infty$ have the following easily calculable upper bound:
\begin{equation}
    \overline{\mathcal{I}_1} (\{M_i\}) := \max_{\ket{\Gamma}}  \sum_i \frac{\max\{I(R:A)_i,I(R:B)_i\}}{2^{\max\{I(R:A)_i, I(R:B)_i\}}}.  
\end{equation}

Without too much redundancy in $\{M_i\}_{i\in I}$ i.e. $\sum_i 2^{-I_i} \approx 1$ , $\overline{\mathcal{I}_i}$ becomes a good approximation of $\mathcal{I}_\infty$.

\fi

\begin{theorem} \label{thm:tradeoff} For any $d-$dimensional  universal quantum masking process $\Phi_M$ composed of random bipartite embeddings $\{M_i\}_{i\in I}$ with orthogonal images, the following inequality holds.
\begin{equation} \label{ineq:trade}
    \mathcal{I}_\infty (\{M_i\}_{i\in I}) \leq \mathcal{R}(\Phi_M).
\end{equation}
\end{theorem}

Note that this inequality can be saturated (e.g. quantum one-time pad \cite{mosca2000private} and 4-qubit masker \cite{lie2019unconditionally}.) An important implication of the inequality above, which is highlighted in the difference between na\"{i}ve Theorem  \ref{thm:naive} and more refined Theorem \ref{thm:inq} and \ref{thm:tradeoff}, is that the most relevant property of masking interaction is not the mean information flow, but the mean \textit{evenness of information flow}. For example, for a hiding process in which information entirely flows to the system $A$ with the total probability of $1/2$ and vice versa for the system $B$, the mean information flow for each system is the same. However, the evenness of information flow is at minimum in any subchannel, therefore this process requires at least $2\log d$ bits of randomness for masking $d$-dimensional quantum information.

 One possible issue of randomness usage in quantum information process is the nonuniformness of randomness source. However, typicality of random state \cite{cover2012elements} asserts that one could treat many copies of nonuniform random state as a single uniform random state when one can permit a small error. The robustness of Theorem \ref{thm:cap} guarantees the robustness of the inequalities above, as one can substitute every $I_i$ terms with $I_i - e$ for the non-perfect masking case with the entanglement-assisted classical capacity $e$.  Therefore the results shown here are compatible with other analyses on random quantum processes that use uniform randomness exclusively \cite{lie2019unconditionally,boes2018catalytic}. 

It is impossible to delete quantum information \cite{pati2000impossibility}. Therefore, to hide information from one system, unless one just displaces quantum information to another system, one needs to `cancel out' the information by randomly altering the information. What theorem \ref{thm:cap} says is that large amount of information leakage requires large amount of randomness to conceal it. From this one might speculate that the presence of a single `large' information leakage subchannel may require large amount  of randomness to scramble it, i.e.,
\begin{equation}
   \max_{X\in\{A,B\},i} I(R:X)_i \stackrel{?}{\leq} \mathcal{R}(\Phi_M).
\end{equation}
The following example, however, disproves this speculation. A subchannel with high channel capacity need not be cancelled by other highly randomized subchannels with equally high channel capacity, if partial distribution of quantum information to two parties is allowed. Note that we still have $\min_i \max_{X\in\{A,B\}} I(R:X)_i \leq \mathcal{R}(\Phi_M)$ nonetheless.

\textbf{(Counter) Example}
Consider the following families of bipartite embeddings $\mathcal{H}_I \to \mathcal{H}_A \otimes \mathcal{H}_B$, defined on an $d$-dimensional Hilbert space $\mathcal{H}_I$ with \textit{odd} number $d$. For an input state $\ket{\psi}_I=\sum_{i=1}^d \alpha_i \ket{i}_I$ with fixed bases $\{\ket{i}_X\}$ for each $X \in \{A,B,I\}$,
$M_{A,i} \ket{\psi}_I := (Z^i \ket{\psi}_A) \otimes \ket{i+d}_B,$
$M_{B,i} \ket{\psi}_I := \ket{i+d}_A \otimes (Z^i \ket{\psi}_B), $
for $1 \leq i \leq d$ and
$M_j \ket{\psi}_I := \sum_i \alpha_{i} \ket{i \oplus j}_A \otimes \ket{i \oplus 2j}_B,$
where $\oplus$ stands for modular sum modulo $d$ for $1 \leq j \leq d-1$.

One can observe that $M_{A,i}$ and $M_{B,i}$ are bipartite embeddings with completely uneven information flow but $M_j$ has even flow. These bipartite embeddings with a safe state $\sigma_S := \sum_{i=1}^d \frac{1}{d(d+1)}\dyad{A,i}_S + \sum_{i=1}^d \frac{1}{d(d+1)}\dyad{B,i}_S + \sum_{j=1}^{d-1} \frac{1}{d+1} \dyad{j}_S$ form a universal quantum masker with $\mathcal{R}= \log (d+1) + \frac{2}{d+1} \log d$. Although $\max_i I(R:X)_i = 2\log d$, we have $\mathcal{R} < 2\log d$ for all $d\geq 2$. However, the bound above almost sharply captures the randomness cost with the small gap of $\mathcal{R}-{\mathcal{I}_1}= \log(1+d^{-1}) - \frac{d-1}{d(d+1)}\log d$ that approaches zero as $d\to \infty$. 

%therefore $\mathcal{R} = \mathcal{I}_\infty$ by regularizing both sides.

\textbf{Sharing quantum secret without attending} For every quantum masker $\Phi_M$, its safe state $\sigma_S$'s purification system $K$, (meaning that there exists a bipartite pure state $\ket{\Sigma}_{SK}$ such that $\Tr_K \dyad{\Sigma}_{SK}=\sigma_S$), automatically shares its share of quantum secret generated from a (2,3)-threshold quantum secret sharing scheme \cite{cleve1999share}, without interacting with either of the systems $S$ and $I$. Here the quantum secret is the quantum information that has been masked. This statement means that either of two groups of parties, $AK$ or $BK$, can restore the quantum information that was masked without help of the other party. This is direct result of the no-hiding theorem \cite{braunstein2007quantum}, which states that if the quantum information is hidden from one party, then it should be isometrically transferred to the remaining parties.

In the quantum masking scenario, if the quantum information is hidden from, say, the system $A$, it is isometrically transferred to the systems $BK$, which allows them to restore the quantum information directly. This observation implies that initially distributed entanglement has ability to transfer a share of quantum secret generated at the later point of time, and the inequality derived in this work gives the lower bound on the amount of entanglement in terms of property of interaction used in the masking process. This observation yields an insight on the process of black hole evaporation discussed in the next section.

We remark that this observation allows us to estimate the sizes of unauthorized sets of not only $(2,3)$-threshold quantum secret sharing protocol but also any pure $(k,2k+1)$-threshold protocol with our result. This follows from the observation that partitioning $2k+1$ parties into any 3 unauthorized sets  yields $(2,3)$-threshold quantum secret sharing protocol. This lower bound is stronger than the previously known $\log d$ bound \cite{imai2005information} and admits estimation of sizes of unauthorized systems for nonperfect quantum secret sharing protocols from the robustness against error.

    \textbf{Black hole evaporation} Hawking's semi-classical analysis \cite{hawking1974black,hawking1975particle} of the outgoing radiation from black hole indicates that the flow of particles from a black hole should contain no information related to in-fallen matter. Plus to this, information cannot be stored in the black hole because it can be vanished at the final stage of its evaporation. This is impossible unless quantum information can be lost, which is forbidden by the unitarity of quantum mechanics. From these observations one could conclude that the information of in-fallen matter should be `masked' into the correlation between the black hole and the Hawking radiation thereof. However, the no-masking theorem says that it is impossible if the masking process is unitary. This paradox is called the \textit{black hole information paradox}.

A possible resolution of black hole information paradox is the fast scrambling in black hole \cite{hayden2007black}. Once a black hole starts in a pure quantum state and if the time evolution after its creation is unitary, then the black hole's internal state is always entangled with the Hawking radiation it has emitted. It is believed that after the Page time \cite{page1993information} of the given black hole, the black hole's internal state is nearly maximally entangled with all the Hawking radiation that has been radiated up to that point. The analysis in \cite{hayden2007black} suggests that in this scenario any $k$ qubits of information falling into the black hole can be almost perfectly retrieved by an observer who has access to the every Hawking radiation emitted from the black hole by acquiring just a little bit of more Hawking radiation up to $k+c$ qubits where $c$ is a constant that only depends on the desired error rate. In other words, black holes function as mirrors in the model.

The model depends on the \textit{ad hoc} assumption of typical unitary evolution under Haar distribution of black hole internal state and the assumption that its entire surface participates in the scrambling process. This assumption, however, is not indispensable \cite{wakakuwa2019one}. Our result here gives a way to examine the consistency between assumptions on scrambling and evaporation processes.

Here we suppose the internal interaction and evaporation process of black hole as a quantum masking process since after the emission of $k$ qubits through Hawking radiation, either the black hole internal state or the just radiated $k$ qubits of emission should not have any information on the $k$ qubits of in-falling object. Suppose that for a given moment parametrized by time $T$ of a black hole's lifetime, the internal interaction is determined by unitary $M$ (with the accompanying set of bipartite embeddings $\{M_i:=M (I \otimes \ket{i})\}_{i\in I}$, ) from a hypothetical theory on the black hole dynamics. From this, one can calculate $\mathcal{I}_\infty (\{M_i\}_{i\in I})$. Also if the entropy of entanglement of the black hole (whose contribution is dominant in the thermodynamic entropy of the black hole \cite{braunstein2013better} without firewall \cite{almheiri2013black}) for the given moment is $S(T)$, and only $c(t^*)$ fraction $(0\leq c \leq 1)$ of the black hole surface can participate for the given scrambling time $t^*$, then we have the following relation.

\begin{equation}
    \mathcal{I}_\infty (\{M_i\}_{i\in I}) \leq c(t^*)S(T)
\end{equation}

Even without explicit calculation of $\mathcal{I}_\infty (\{M_i\}_{i\in I})$, one already has the relation between the reflection capacity of the black hole and its entropy of entanglement at the moment from the trivial lower bound of $k \leq \mathcal{I}_\infty (\{M_i\}_{i\in I})$ where $k$ is the number of qubits that can be scrambled at a time with given scrambling time $t^*$. In other words, in the early or late stage of black hole's evolution, (i.e. $S$ is small) the black hole should have very small reflection capacity. This relation provides a way to check consistency between quantities determined from independent theories such as scrambling time, internal evolution of black hole and time evolution of the entropy of entanglement of black hole. Especially, when the internal interaction of black hole is inherently asymmetric, (e.g. information only `heads' outward through radiation as it is proposed in \cite{braunstein2013better} as quantum one-time pad encoding of in-fallen information or information always coherently `falls' into the horizon for each bipartite embedding) the upper bound of the reflection capacity of black hole can drop up to the half of the entropy of entanglement of the black hole.

\textit{Note :} We recently learned of the independent result of Feng Ding and Xueyuan Hu \cite{FengHu2019} on counterexamples of Conjecture 1 for qutrit unitary quantum masker. We remark that the counterexample in this work forms a different family of quantum states from that of \cite{FengHu2019} as it is for $d\geq 4$ dimensional unitary quantum maskers.

\section{APPENDIX : PROOF OF THE RESULTS}

\begin{thm} \label{thm_a:naive}
For any universal quantum masker $\Phi_M$,
\begin{equation*} \label{naive}
    \max_{X\in\{A,B\}}\sum_i p_i I(R:X)_i \leq \mathcal{R}(\Phi_M).
\end{equation*}

\end{thm}
\begin{proof}
Without loss of generality, we can assume $X=A$. Then as $I(R:A)_i = S(R)_i + S(A)_i - S(B)_i = S(R) + S(A)_i - S(B)_i$, it suffices to prove the following inequality,

\begin{equation}
    S(R) + \sum_i p_i S(A)_i \leq H(\{p_i\}) + \sum_i p_i S(B)_i,
\end{equation}
because $H(\{p_i\})=S(\sigma_S)$, where $H(\{q_i\}):= -\sum_i q_i \log q_i$ is the Shannon entropy of the probability distribution $\{q_i\}$. As $\sum_i p_i S(\rho_i) \leq S(\sum_i p_i \rho_i)$ from the concavity of von Neumann entropy \cite{watrous2018theory}, we have 
\begin{equation}
    S(R) + \sum_i p_i S(A)_i \leq S(R) + S(A) = S(RA),
\end{equation}
because systems $R$ and $A$ are in a product state because of the masking property of $\Phi_M$. If we define $\Phi_{i}^A(\rho) := \Tr_B \big(M(\rho_I\otimes \dyad{i}_S)M^\dag \big)$, $S(RA)$ equals to
\begin{equation}
    \begin{gathered}
    S(\sum_i p_i (\mathds{1}_R\otimes \Phi_{i}^A)(\dyad{\Gamma}_{RI})).
    \end{gathered}
\end{equation}
From the following property of von Neumann entropy \cite{watrous2018theory},
\begin{equation}
    S(\sum_i q_i \rho_i) \leq H(\{q_i\}) + \sum_i q_i S(\rho_i),
\end{equation}
we have
\begin{equation}
    S(RA) \leq H(\{p_i\}) + \sum_i p_i S(RA)_i.
\end{equation}
 This proves the wanted inequality since $S(RA)_i = S(B)_i$ because the systems $RAB$ are in a pure state. Since this inequality holds for arbitrary maximally entangled state $\ket{\Gamma}_{RI}$, we have the wanted result.
\end{proof}

\begin{thm} \label{thm_a:cap}
    Let a convex sum of subchannels $\{\mathcal{N}_i\}$ $\sum_i p_i \mathcal{N}_i = \mathcal{N}$ be an almost-erasure quantum channel, i.e., a quantum map that has entanglement-assisted classical capacity of $e$. Then for all $i$, the entanglement-assisted classical capacity $C_{EA}$ of each masking component is upper bounded by the information content of the randomness source for its corresponding $i$, i.e.,
    \begin{equation} \label{ineq1}
        C_{EA}(\mathcal{N}_i) - e \leq -\log p_i,
    \end{equation}
\end{thm}

\begin{proof}We first prove the case where $\mathcal{N}$ is a complete erasure channel $(e=0.)$
From the masking property, no information should be conveyed over a universal quantum masker to each party. Assume that for some $i$,
\begin{equation}
    C_{EA}(\mathcal{N}_i) > -\log p_i.
\end{equation}
Then there exist $\epsilon, \delta > 0$ for all positive integer $n>0$ such that
\begin{equation}
    p_i^n (1-\delta) > \frac{1}{2^{n(1-\epsilon)C_{EA}(\mathcal{N}_i)}}.
\end{equation}
We now consider an information transfer task between an encoder and a decoder. Encoder uniformly samples an arbitrary letter $L$ from $N$ possible candidates, then encode and transfer it by using the channel $\mathcal{N}$ multiple times with the assistance of unbounded amount of entanglement. Since $\mathcal{N}$ is a completely lossy channel, decoder should have no information at all about the letter $L$. Therefore the probability of decoder to guess the letter $L$ correctly should never exceed $\frac{1}{N}$.

Nevertheless if the encoder and the decoder chose a strategy in which they always assume that the channel is $\mathcal{N}_i$ instead of $\mathcal{N}$, then because of the achievability theorem, with the probability that is no less than $p_i^n(1-\delta)$, the decoder can guess the letter $L$ uniformly sampled from $2^{n(1-\epsilon)C_{EA}(\mathcal{N}_i)}$ possible alphabets, by using the channel $n$ times with sufficiently large $n$. However, as $p_i^n (1-\delta)> 2^{-n(1-\epsilon)C_{EA}(\mathcal{N}_i)}$ from the assumption, this contradicts the previous statement.

Now, for the case of capacity $e$, one can only have up to $2^{ne}$-fold probability enhancement for the letter with the lowest probability of correct guessing in the information transfer task, when using the channel $n$ times, compared to the complete erasure channel case. But still negation of the assumption yields the existence of an index $i$ for which positive $\epsilon$ and $\delta$ exist for any positive integer $n$ such that
\begin{equation}
    p_i^n (1-\delta) > \frac{2^{ne}}{2^{n(1-\epsilon)C_{EA}(\mathcal{N}_i)}}.
\end{equation}
It implies the existence of strategy of Alice and Bob achieving the probability strictly higher than the maximum probability.

\end{proof}

\begin{thm} \label{thm_a:tradeoff} For any $d-$dimensional  universal quantum masking process $\Phi_M$ composed of random bipartite embeddings $\{M_i\}_{i\in I}$ with orthogonal images, the following inequality holds.
\begin{equation} \label{ineq:trade}
    \mathcal{I}_\infty (\{M_i\}_{i\in I}) \leq \mathcal{R}(\Phi_M).
\end{equation}
\end{thm}

\begin{proof}
    We first prove the following seemingly weaker inequality,
    \begin{equation} \label{ineq:oneshot}
        \mathcal{I}_1 (\{M_i\}_{i\in I})\leq \mathcal{R}(\Phi_M) + \frac{2\log d}{d}.
    \end{equation}
    From Theorem 3 of the main text,
    \begin{equation} \label{eqn:mid}
    \log d + \sum_i p_i |S(A)_i - S(B)_i| \leq \mathcal{R}(\Phi_M),
    \end{equation}
    and that $I_i = \log d + |S(A)_i - S(B)_i|$, we have
    \begin{equation} \label{ineq:int}
        \max_{\ket{\Gamma}} \min_{\{p_i\}}\sum_i p_i I_i \leq \mathcal{R}(\Phi_M),
    \end{equation}
    where the maximization is over every bipartite state $\ket{\Gamma}_{RI}$ and the minimization is over every possible probability distribution $\{p_i\}$ that satisfies $p_i \leq 2^{-I_i}$ for each $i$. When $S$ is an index set that saturates the minimization in the definition of $\mathcal{I}_{1} (\{M_i\}_{i \in I})$,
    \begin{equation} \label{eqn:div}
    \mathcal{I}_{1} (\{M_i\}_{i \in I}) = \max_{\ket{\Gamma}} \min_{S \subseteq I} H\left(\left\{\frac{1}{2^{I_i}}\right\}_{i\in S}\right),
    \end{equation}
    the minimization term in (\ref{ineq:int}) with such probability assignment is larger than the Shannon entropy $H(\{2^{-I_i}\}_{i \in S_0})$ with the index set $S_0 := S \setminus \{i_0\}$, because the probability distribution saturating the minimization in (\ref{ineq:oneshot}) is assigning $p_i=2^{-I_i}$ to indices in some $S'$ with the lowest values of $I_i$ until $\sum_{i \in S'} 2^{-I_i} \leq 1$ and assigning the probability $p_{i_0}=1-\sum_{i \in S'} 2^{-I_i}$ to the index with the next smallest $I_{i_0}$ and the fact that $\{2^{-I_i}\}_{i \in S_0}$ is an incomplete probability distribution. Combined with the fact that any term $2^{-I_i} I_i$ is not larger than $2d^{-1}\log d$, this yields the wanted result (\ref{ineq:oneshot}). By noting that $\mathcal{R}(\Phi_M)$ is additive, that is, $\mathcal{R}(\Phi_M^{\otimes n})= n\mathcal{R}(\Phi_M)$, regularizing both sides of (\ref{ineq:oneshot}) yields (\ref{ineq:trade}).
\end{proof}

\bibliography{main}

\end{document}